\documentclass[11pt, letterpaper]{article}
\usepackage{standalone}
\usepackage{etoolbox}

\newbool{sosasubm}
\setbool{sosasubm}{false}

\usepackage[letterpaper,margin=1in]{geometry}
\usepackage{amsthm}
\usepackage{amsmath}
\usepackage{amssymb}
\usepackage{thmtools}
\usepackage{thm-restate}
\usepackage{times}
\usepackage{mathtools}
\usepackage{comment}

\usepackage{setspace}

\usepackage{placeins}

\usepackage[table]{xcolor}

\definecolor{linkcol}{rgb}{0,0,0.38}
\definecolor{citecol}{rgb}{0.8,0,0}
\definecolor{urlcol}{rgb}{0.1,0.35,0}

\usepackage[pdfpagelabels,bookmarks=false,hyperfootnotes=false]{hyperref}
\hypersetup{colorlinks, linkcolor=linkcol, citecolor=citecol, urlcolor=urlcol}

\DeclareFontFamily{U}{BOONDOX-calo}{\skewchar\font=45 }
\DeclareFontShape{U}{BOONDOX-calo}{m}{n}{
  <-> s*[1.05] BOONDOX-r-calo}{}
\DeclareFontShape{U}{BOONDOX-calo}{b}{n}{
  <-> s*[1.05] BOONDOX-b-calo}{}
\DeclareMathAlphabet{\link}{U}{BOONDOX-calo}{m}{n}
\DeclareMathAlphabet{\blink}{U}{BOONDOX-calo}{b}{n}

\usepackage[utf8]{inputenc}

\usepackage[
backend=biber,
style=alphabetic,
citestyle=alphabetic,
maxalphanames=4,
maxcitenames=99,
mincitenames=98,
maxbibnames=99,
giveninits=true,
]{biblatex}

\usepackage{xpatch}

\makeatletter
\patchcmd\blx@bblinput{\blx@blxinit}
                      {\blx@blxinit
                      }{}{\fail}
\makeatother

\usepackage[nameinlink]{cleveref}

\usepackage{lmodern}

\usepackage{bbm}
\usepackage{thm-restate}

\newtheoremstyle{light} {\topsep}                    {\topsep}                    {\itshape}                   {}                           {\scshape}                   {.}                          {.5em}                       {}  

\newtheorem{theorem}{Theorem}[section]
\newtheorem{lemma}[theorem]{Lemma}

\newtheorem{proposition}[theorem]{Proposition}
\newtheorem{definition}[theorem]{Definition}

\newtheorem{corollary}[theorem]{Corollary}
\newtheorem{conjecture}[theorem]{Conjecture}

\theoremstyle{light}

\makeatletter
\if@cref@capitalise
\crefname{claiminproof}{Claim}{Claims}
\else
\crefname{claiminproof}{claim}{claims}
\fi

\if@cref@capitalise
\crefname{algocf}{Algorithm}{Algorithms}
\else
\crefname{algocf}{algorithm}{algorithms}
\fi

\if@cref@capitalise
\crefname{conjecture}{Conjecture}{Conjectures}
\else
\crefname{conjecture}{conjecture}{conjectures}
\fi

\if@cref@capitalise
\crefname{thm}{Theorem}{Theorems}
\else
\crefname{thm}{theorem}{theorems}
\fi

\if@cref@capitalise
\crefname{lem}{Lemma}{Lemmas}
\else
\crefname{lem}{lemma}{lemmas}
\fi

\makeatother

\usepackage{url}
\usepackage{array}

\usepackage{setspace}

\usepackage{wrapfig}

\makeatletter
\newcommand{\labeltarget}[1]{\Hy@raisedlink{\hypertarget{#1}{}}}
\makeatother

\usepackage{upgreek}

\usepackage{complexity}

\usepackage[inline]{enumitem}
\usepackage{moreenum}

\usepackage[super]{nth}

\usepackage{bm}

\usepackage{xspace}

\usepackage{ifthen}

\usepackage[immediate]{silence}
\usepackage{sectsty}
\DeactivateWarningFilters[tmp]
\allsectionsfont{\boldmath}

\setlist[enumerate]{nosep,topsep=0.1em}
\setlist[enumerate,1]{label=(\roman*), leftmargin=2.2em}
\setlist[itemize]{nosep,topsep=0.3em}

\usepackage[textsize=footnotesize, color=blue!30!white]{todonotes}
\usepackage{xfrac}

\usepackage{tikz}
\usetikzlibrary{calc}
\usetikzlibrary{math}
\usetikzlibrary{shapes.geometric}
\usetikzlibrary{arrows}
\usetikzlibrary{decorations.pathreplacing, decorations.pathmorphing}

\usepackage{tcolorbox}
\tcbuselibrary{skins}

\usepackage{float}
\usepackage{graphicx}
\graphicspath{{../graphics/}}
\makeatletter
\newcommand\appendtographicspath[1]{\g@addto@macro\Ginput@path{#1}}
\makeatother

\usepackage[margin=10pt,font=small,labelfont=bf,skip=-5pt]{caption}
\usepackage{subcaption}

\usepackage[linesnumbered,vlined,ruled,algo2e]{algorithm2e}
\SetAlgoSkip{bigskip}
\SetAlgoInsideSkip{smallskip}

\usepackage{mdframed}

\let\truehypersetup\hypersetup
\renewcommand\hypersetup[1]{}
\usepackage{bigfoot}
\let\hypersetup\truehypersetup
\DeclareMathOperator{\argmin}{argmin}

\DeclareMathOperator{\load}{load}

\renewcommand{\epsilon}{\varepsilon}

\renewcommand{\NP}{\text{NP}}

\makeatletter
\let\@@pmod\pmod
\DeclareRobustCommand{\pmod}{\@ifstar\@pmods\@@pmod}
\def\@pmods#1{\mkern8mu({\operator@font mod}\mkern 6mu#1)}
\makeatother

\makeatletter
\let\@@mod\mod
\DeclareRobustCommand{\mod}{\@ifstar\@mods\@@mod}
\def\@mods#1{\mkern8mu{\operator@font mod}\mkern 6mu#1}
\makeatother

\definecolor{green}{rgb}{0.4,0.85,0.6}

\title{Unsplittable Cost Flows from Unweighted Error-Bounded Variants}

\ifbool{sosasubm}{
  \author{}
}{
\author{
Chaitanya Swamy\thanks{
Dept.~of Combinatorics \& Optimization, Univ.~Waterloo, Waterloo, Canada.
Email: \href{mailto:cswamy@uwaterloo.ca}{cswamy@uwaterloo.ca}.
}
\and
Vera Traub\thanks{
Department of Computer Science, ETH Zurich, Zurich, Switzerland.
Email: \href{mailto:vtraub@ethz.ch}{vtraub@ethz.ch}.
}
\and
Laura Vargas Koch\thanks{
RWTH Aachen University, Aachen, Germany.
Email: \href{mailto:vargaskoch@gdm.rwth-aachen.de}{vargaskoch@gdm.rwth-aachen.de}
}
\and
Rico Zenklusen\thanks{
Department of Mathematics, ETH Zurich, Zurich, Switzerland.
Email: \href{mailto:ricoz@ethz.ch}{ricoz@ethz.ch}.}
}
}
\date{}

\begin{document}

\maketitle
\thispagestyle{empty}
\addtocounter{page}{-1}

\begin{abstract}
A famous conjecture of Goemans on single-source unsplittable flows states that one can turn any fractional flow into an unsplittable one of no higher cost, while increasing the load on any arc by at most the maximum demand.
Despite extensive work on the topic, only limited progress has been made.
Recently, Morell and Skutella suggested an alternative conjecture, stating that one can turn any fractional flow into an unsplittable one without changing the load on any arc by more than the maximum demand.

We show that their conjecture implies Goemans' conjecture (with a violation of twice the maximum demand).
To this end, we generalize a technique of Linhares and Swamy, used to obtain a low-cost chain-constrained spanning tree from an algorithm without cost guarantees.
Whereas Linhares and Swamy's proof relies on Langrangian duality, we provide a very simple elementary proof of a generalized version, which we hope to be of independent interest.
Moreover, we show how this technique can also be used in the context of the weighted ring loading problem, showing that cost-unaware approximation algorithms can be transformed into approximation algorithms with additional cost guarantees.
\end{abstract}
 
\section{Introduction}
\label{sec:intro}

Network flow problems are among the most basic and ubiquitous problems in combinatorial optimization, both in theory and practice.
(We refer the interested reader to the classical monograph \cite{ahujaNetworkFlowsTheory1993} and the more recent textbook \cite{williamsonNetworkFlowAlgorithms2019} on the topic.)
In the classical maximum flow problem, flow can split and merge at nodes.
However, in many natural settings, this is not possible, which naturally leads to the notion of unsplittable flows.
In an unsplittable flows, the complete flow from a source to a destination needs to be sent along a single path.

One of the most heavily studied questions in this context is the single-source unsplittable flow problem (SSUF), originally introduced by Kleinberg~\cite{kleinbergApproximationAlgorithmsDisjoint1996,kleinbergSingleSourceUnsplittableFlow1996}.
In SSUF, we are given a directed graph $(V,A)$ with a single source $s\in S$, and several sinks $T=\{t_1, \ldots, t_k\} \subseteq V$.
Additionally, each sink $t\in T$ comes with a demand $d_t\in \mathbb{R}_{\geq 0}$.
Without loss of generality, we assume that the sinks are distinct.
We refer to the tuple $G=(V,A,s,T,d)$ as an \emph{SSUF network}.
Moreover, if each arc $a\in A$ has additionally a nonnegative cost $c(a)\in \mathbb{R}_{\geq 0}$, then we call the tuple $G=(V,A,s,T,d,c)$ a \emph{weighted SSUF network}.

In their canonical form, unsplittable flow problems ask to find, for a given SSUF network $(V,A,s,T,d)$ with arc capacities $u: A \to \mathbb{R}_{\geq 0}$, an $s$-$t$ path $P_t \subseteq A$ for each $t\in T$ so that
\begin{equation*}
  \sum_{t\in T: a\in P_t} d_t \leq u(a) \qquad \forall a\in A.
\end{equation*}
In other words, when sending, for each $t\in T$, a flow of $d_t$ units from $s$ to $t$ along the path $P_t$, the flow on each arc $a\in A$ does not exceed its capacity $u(a)$.
In weighted/cost versions of the problem, one is additionally interested in minimizing the total cost $\sum_{t\in T} c(P_t)$ of the unsplittable flow.
For brevity, given an unsplittable flow $\mathcal{P}=(P_t)_{t\in T}$ in an SSUF network $G$ (weighted or unweighted), we denote by $f^{\mathcal{P}}\in \mathbb{R}^A_{\geq 0}$ the corresponding arc-flow, i.e.,
\begin{equation*}
  f^{\mathcal{P}}(a) \coloneqq \sum_{t\in T: a\in P_t} d_t \qquad \forall a\in A.
\end{equation*}

Even just deciding whether an unsplittable flow exists is \NP-hard.
However, Dinitz, Garg, and Goemans~\cite{dinitzSingleSourceUnsplittableFlow1999} showed that, interestingly, if we allow for a slight violation of the capacities, determining an unsplittable flow becomes computationally tractable.
\begin{theorem}[\cite{dinitzSingleSourceUnsplittableFlow1999}]\label{thm:dinitzMain}
  Let $(V,A,s,T,d)$ be an SSUF network, and let $x\in \mathbb{R}^A_{\geq 0}$ be a fractional flow.
  Then there is a polynomial-time algorithm that finds an unsplittable flow $\mathcal{P}$ satisfying
  \begin{equation*}
     f^{\mathcal{P}}(a) \leq x(a) + d_{\max} \qquad \forall a\in A,
  \end{equation*}
  where $d_{\max} = \max_{t\in T} d_t$.
\end{theorem}
Note that, when given an SSUF network $(V,A,s,T,d)$ with capacities $u\colon A \to \mathbb{R}_{\geq 0}$, the above statement implies the following natural implication.
If there exists a fractional flow respecting the capacities $u$, then one can, in polynomial time, find an unsplittable flow violating each capacity by at most $d_{\max}$.
Indeed, to obtain this statement from \Cref{thm:dinitzMain}, one first computes a fractional flow $x$ that respects the capacities, and then applies \Cref{thm:dinitzMain} to find an unsplittable flow $(P_t)_{t\in T}$ with $f^{\mathcal{P}}(a) \leq x(a) + d_{\max} \leq u(a) + d_{\max}$ for all $a\in A$.
Finally, if no fractional flow exists, then there is also no unsplittable flow that obeys the capacities.

A famous conjecture of Goemans claims that \Cref{thm:dinitzMain} can be generalized to the weighted case as follows (see~\cite{skutellaApproximatingSingleSource2002}).
\begin{conjecture}[Goemans]\label{conj:goemans}
  Let $(V, A, s, T, d, c)$ be a weighted SSUF network, and let $x\in \mathbb{R}^A_{\geq 0}$ be a fractional flow.
  Then there is a polynomial-time algorithm that finds an unsplittable flow $\mathcal{P}$ satisfying 
  \begin{itemize}
    \item $f^{\mathcal{P}}(a) \leq x(a) + d_{\max} \quad \forall a\in A$, and
    \item $c^T f^{\mathcal{P}} \leq c^T x$.
  \end{itemize}
\end{conjecture}

Note that in \Cref{conj:goemans}, one can assume without loss of generality that the underlying graph $(V,A)$ is acyclic.
Indeed, for a general (not necessarily acyclic) graph, we can first simplify the fractional flow $x$ by eliminating flow on cycles, and then delete all arcs not carrying any fractional flow.
This leads to a new fractional flow $x' \leq x$ on an acyclic graph, and any unsplittable flow satisfying the requirements of the conjecture for $x'$ will also satisfy them for $x \geq x'$.

Let $Q_G\subseteq \mathbb{R}^A_{\geq 0}$ be the polytope of all fractional (arc) flows, i.e.,
\begin{equation*}
  Q_G \coloneqq \left\{x \in \mathbb{R}_{\geq 0}^A \colon \text{ for $v\in V$ we have } x(\delta^+(v)) - x(\delta^-(v)) =
  \begin{cases}
    \sum_{t\in T} d_t &\text{if }v=s, \\
    -d_v &\text{if }v\in T, \\
    0 &\text{if }v \in V\setminus (\{s\}\cup T)
  \end{cases}
  \right\}.
\end{equation*}

Hence, using this notation, \Cref{thm:dinitzMain} says that given a point $x \in Q_G$, one can obtain in polynomial time an unsplittable flow $\mathcal{P}= (P_t)_{t\in T}$ with $f^{\mathcal{P}}(a) \leq x(a) + d_{\max}$ for all $a\in A$.

Recently, Morell and Skutella~\cite{morellSingleSourceUnsplittable2022} suggested two additional conjectures, a weaker and a stronger version, to strengthen \Cref{thm:dinitzMain}.
In the weaker version, stated below, instead of requiring a cost guarantee as in \Cref{conj:goemans}, the conjecture requires additional lower bounds on the flow values.
Contrary to \Cref{conj:goemans}, in \Cref{conj:morellLowerBounds} below it is important to {\em impose} that the network is acyclic (see the discussion in \cite{traubSingleSourceUnsplittableFlows2024}). 
\begin{conjecture}[\cite{morellSingleSourceUnsplittable2022}]\label{conj:morellLowerBounds}
Let $G=(V,A,s,T,d)$ be an acyclic SSUF network and let $x\in Q_G$.
Then one can compute in polynomial time an unsplittable flow $\mathcal{P}$ in $G$ with
\begin{equation*}
  x(a) - d_{\max} \leq f^{\mathcal{P}}(a) \leq x(a) + d_{\max} \qquad \forall a\in A.
\end{equation*}
\end{conjecture}

Moreover, Morell and Skutella~\cite{morellSingleSourceUnsplittable2022} also conjectured a stronger version of \Cref{conj:morellLowerBounds}, stated below, which imposes, in addition to lower and upper bounds, also cost guarantees.
\begin{conjecture}[\cite{morellSingleSourceUnsplittable2022}]\label{conj:morellStrong}
Let $G=(V,A,s,T,d)$ be an acyclic SSUF network and let $x\in Q_G$.
Then one can compute in polynomial time an unsplittable flow $\mathcal{P}$ in $G$ satisfying
\begin{itemize}
  \item $c^T f^{\mathcal{P}} \leq c^T x$, and
  \item $x(a) - d_{\max} \leq f^{\mathcal{P}}(a) \leq x(a) + d_{\max} \qquad \forall a\in A$.
\end{itemize}
\end{conjecture}

We will show in the following that \Cref{conj:morellLowerBounds} allows for deriving the statement mentioned below, which yields a slightly weaker version of \Cref{conj:morellStrong} (and therefore also of \Cref{conj:goemans}), where the capacity violation can be $2 d_{\max}$ instead of $d_{\max}$.

\begin{conjecture}\label{conj:costImplication}
Let $G=(V,A,s,T,d,c)$ be a weighted, acyclic SSUF network and let $x\in Q_G$.
Then one can compute in polynomial time an unsplittable flow $\mathcal{P}$ in $G$ with
\begin{itemize}
  \item $c^T f^{\mathcal{P}} \leq c^T x$, and
  \item $x(a) - 2d_{\max} \leq f^{\mathcal{P}}(a) \leq x(a) + 2 d_{\max} \quad \forall a\in A$.
\end{itemize}
\end{conjecture}

Hence, formally, we will prove the following.
\begin{theorem}\label{thm:mainUnsplitReduction}
  If \Cref{conj:morellLowerBounds} holds, then so does \Cref{conj:costImplication}.
\end{theorem}

We note that it remains wide open whether Goemans' conjecture holds even with a weaker capacity violation of $O(d_{\max})$, and it seems fair to say that proving such a weaker form of Goemans' conjecture would already be considered a breakthrough in the field.
Hence, \Cref{thm:mainUnsplitReduction} shows in particular that even \Cref{conj:morellLowerBounds} implies this slightly weaker form of Goemans' conjecture, and is therefore a more ambitious statement to prove (if true).
Moreover, \Cref{thm:mainUnsplitReduction} also implies that the two conjectures of Morell and Skutella~\cite{morellSingleSourceUnsplittable2022} are almost equivalent, up to a violation of $2d_{\max}$ versus $d_{\max}$.

\medskip

We show \Cref{thm:mainUnsplitReduction} by building upon a technique of Linhares and Swamy~\cite{linharesApproximatingMincostChainconstrained2018}, originally introduced in the context of chain-constrained spanning trees.
This is a technique to transform certain types of cost-unaware approximation algorithms into ones with additional cost guarantees.
In particular, we provide a very simple and elementary way to analyze a generalized version of this technique, which can then be used to derive \Cref{thm:mainUnsplitReduction}.

To underline the generality of the approach, we present a second application, to the weighted ring loading problem, which is an unsplittable flow problem on a cycle.
More precisely, in the weighted ring loading problem, we are given an undirected cycle $C=(V,E)$ with source-sink pairs $s_i,t_i \in V$ together with a demand $d_i\in \mathbb{R}_{\geq 0}$, for $i\in [k]$. 
The task is to route each demand $d_i$ from $s_i$ to $t_i$ along one path $P_i\subseteq E$ of the two $s_i$-$t_i$ path along the cycle $C$, so that the maximum flow on any edge is as small as possible.
This problem has been extensively studied in the literature without costs.
We show that our technique does extend to this setting, and when applied to a recent result of Däubel~\cite{daubel2022improved}, which shows how to get a solution violating each capacity by at most $\frac{13}{10} d_{\max}$ if there is a fractional solution, we can obtain the following result with a cost guarantee.
\begin{theorem}\label{thm:mainRingLoading}
	Given a weighted ring loading instance with non-negative costs and a fractional solution $x$. There is a polynomial time algorithm that
	computes an unsplittable solution $\mathcal{P}=(P_i)_{i \in [k]}$, such that
	\begin{itemize}
	  \item $\displaystyle\sum_{i=1}^k c(P_i) \leq c^T x  $, and
		\item $\displaystyle\sum_{\substack{i\in [k]\colon e\in P_i}} d_i \leq x_e+ \frac{13}{5} \, d_{\max} $ for all $e\in E$.
	\end{itemize}
\end{theorem}

\subsection{Further related work}

Recently, Majthoub Almoghrabi, Skutella, and Warode~\cite{majthoub2025integer} proved \Cref{conj:morellStrong} for multi-commodity instances on series-parallel networks. 
Traub, Vargas Koch, and Zenklusen~\cite{traubSingleSourceUnsplittableFlows2024} showed that \Cref{conj:costImplication} holds in the single source setting on planar graphs.
Previously, Skutella~\cite{skutellaApproximatingSingleSource2002} proved \Cref{conj:goemans} for the special case when the demands are all multiples of each other, and Morell and Skutella~\cite{morellSingleSourceUnsplittable2022}  proved \Cref{conj:morellStrong} (and thus \Cref{conj:morellLowerBounds}) for the same special case.

The unweighted, multi-commodity setting on outerplanar networks is a natural generalization of the ring loading problem. Here, Shapley and Shmoys~\cite{shapley2024small} give additive bounds (depending on the graph structure) on how much an unsplittable flow has to exceed a given fractional flow that satisfies all demands. Their result was improved by Alem\'an-Espinosa and Kumar~\cite{aleman2025unsplittable}, who provide a constant bound just depending on the best rounding algorithm of the ring loading problem.

Beside the mentioned recent progress, there has been extensive work on unsplittable flows in a variety of settings, including with different objective functions.
For more information, we refer to the survey of Kolliopoulos~\cite{kolliopoulos2007edge} and to the discussion and references in Grandoni, Mömke, and Wiese~\cite{grandoni2022ptas}.

\subsection{Organization of the paper}

We discuss our approach in \Cref{sec:approach} in a generic way, using the single-source unsplittable flow problem as a running example.
In \Cref{sec:tradeoff}, we will show how our approach can be parameterized to obtain a trade-off between cost and constraint/capacity violations.
Even though this parameterization includes our main unparameterized result as a special case, we first present the unparameterized version at the beginning of \Cref{sec:approach} to keep the presentation simple.
Finally, \Cref{sec:applications} discusses the application of our approach to the single-source unsplittable flow problem and the weighted ring loading problem, and proves in particular \Cref{thm:mainUnsplitReduction} and \Cref{thm:mainRingLoading}.
 \section{Our approach}
\label{sec:approach}

We follow the approach of Linhares and Swamy~\cite{linharesApproximatingMincostChainconstrained2018}, extend it to our needs and present a simplified analysis of it, which, we believe is of independent interest.
Because our simplified analysis applies to a large set of problems, similarly to \cite{linharesApproximatingMincostChainconstrained2018}, we first describe the type of problem settings that we can address with it in a general way, using unsplittable flows only as an example.

Let $N$ be a finite set, and let $Q \subseteq \mathbb{R}^N$ be a polyhedron, which can be thought of as a (linear) relaxation of the problem we want to solve.
The actual solutions we want to optimize over are only a subset $Z\subseteq Q$ of $Q$.
In combinatorial optimization, $Q$ is often the convex hull of the incidence/characteristic vectors of feasible solutions, and $Z$ are the vertices of $Q$, which therefore naturally correspond to the feasible solutions.
This is also the setting in which the results of \cite{linharesApproximatingMincostChainconstrained2018} have been presented.
However, interestingly, for the unsplittable flow setting that we consider, we want to deviate from this classical viewpoint, and observe that the approach presented in \cite{linharesApproximatingMincostChainconstrained2018} extends to this broader setting where $Z$ need not be the vertices of $Q$.
More precisely, when given a weighted SSUF instance $G=(V,A,s,T,d,c)$, we will use $Q=Q_G$ (i.e., the set of all fractional arc flows satisfying the demands), and consider as solutions $Z$ the set
\begin{equation}\label{eq:defZForUnsplit}
  Z_G \coloneqq \left\{ f^{\mathcal{P}} \colon \mathcal{P}=(P_t)_{t\in T} \text{ is an unsplittable flow in $G$}\right\}.
\end{equation}
Note that therefore, for unsplittable flows, $Z_G$ are in general not the vertices of $Q_G$.\footnote{For a small example illustrating this, consider a graph $G=(V,A)$ with $V=\{s,t\}$ and two parallel arcs from $s$ to $t$.
Let there be two commodities, each with a demand of $1$, with source $s$ and sink $t$.
Then $Q_G=\{(x_1,x_2) \in \mathbb{R}_{\geq 0}^A \colon x_1+x_2=2\}$.
However, the unsplittable flow that sends a flow of one along each arc yields the vector $(1,1) \in Z_G$, which is not an extreme point of $Q_G$.}

Moreover, we need what Linhares and Swamy~\cite{linharesApproximatingMincostChainconstrained2018} call a \emph{face-preserving rounding algorithm} (FPRA).
An FPRA is a procedure that, given a point $x\in Q$, returns a solution lying on the minimal face of $Q$ containing $x$.
We use the following extended definition of FPRA, which is explicit about the sets $Q$ and $Z$, and captures its guarantees through what we call an error body, which generalizes the rounding errors considered in \cite{linharesApproximatingMincostChainconstrained2018}.
\begin{definition}[$(Q,Z)$-face-preserving rounding algorithm ($(Q,Z)$-FPRA) with error body $R$]
Let $Q\subseteq \mathbb{R}^N$ be a polyhedron (called \emph{relaxation}), let $Z\subseteq Q$ (called \emph{feasible solutions}), and let $R\subseteq \mathbb{R}^N$ (\emph{error body}) be a convex set containing the origin.
\emph{A $(Q,Z)$-face-preserving rounding algorithm (FPRA) with error body $R$} takes as input a point $x\in Q$ and returns a solution $z\in Z$ with $z-x\in R$ that lies on the minimal face of $Q$ containing $x$.
\end{definition}

Note that for a $(Q,Z)$-FPRA to exist, we need that $Z$ contains in particular all vertices of $Q$.
Indeed, if $x$ is a vertex of $Q$, then the minimal face of $Q$ containing $x$ is just the point $x$ itself, and thus the FPRA must return $z=x$.
We also highlight that the error body $R$ need not be symmetric.

Being an FPRA is a natural property of many rounding procedures.
Moreover, a key observation is that in the single-source unsplittable flow setting, for many of its variants, one can easily transform any rounding procedure into an FPRA due to the following.
Consider, for example, an algorithm resolving \Cref{conj:morellLowerBounds}.
Such an algorithm, which takes a point $x\in Q_G$ and returns an unsplittable flow $\mathcal{P}$ fulfilling the criteria of the conjecture, can be interpreted as a rounding procedure that rounds $x$ to $f^{\mathcal{P}}$.\footnote{Note that it is important that this procedure does not only return $f^{\mathcal{P}}$, which is a point in $Z$, but also the actual path flows given by $\mathcal{P}=(P_t)_{t\in T}$, as they are the unsplittable flow.
However, throughout this paper we will not always repeat this fact and think of points in $Z$ as actual solutions.}
To turn this (rounding) algorithm into a $(Q_G,Z_G)$-FPRA with error body $R=[-d_{\max}, d_{\max}]^A$, we simply first delete all arcs $a$ from $G$ with $x(a)=0$ before applying the algorithm.
This guarantees that we obtain an unsplittable flow $\mathcal{P}$ that does not use any arc with zero $x$-value, which implies that $f^{\mathcal{P}}$ is on the minimal face of $Q_G$ on which $x$ lies.
Note that this crucially exploits that the faces of $Q_G$ are defined by tight nonnegativity constraints, as these are the only inequality constraints of $Q_G$.

The idea is to transform a $(Q,Z)$-FPRA with error body $R$, which is an algorithm oblivious to potential costs of the underlying problem, to an algorithm with cost guarantees.
Hence, we assume that some linear costs $c\colon Q\to \mathbb{R}$ are given that we seek to minimize.
The algorithm we use to obtain cost guarantees from a $(Q,Z)$-FPRA with error body $R$ is presented in~\Cref{algo:main}, where $x-R$ denotes the Minkowski sum of $\{x\}$ and $-R$, i.e., $x-R=\{x-r \colon r\in R\}$.

\begin{algorithm2e}[ht]
  \caption{From a $(Q,Z)$-FPRA with error body $R$ to cost-aware rounding.}\label{algo:main}
  Find an optimal solution $y^*$ to $\argmin\{c^T y \colon y \in Q \cap (x - R) \}$.\label{algline:solveAuxProb}

  Apply the $(Q,Z)$-FPRA to $y^*$ to obtain a solution $z$.\label{algline:applyFPRA}

  \Return $z$.
\end{algorithm2e}

Note that the optimization problem to be solved in line~\ref{algline:solveAuxProb} of \Cref{algo:main} asks to optimize a linear function over a convex set.
To keep the framework general, we have, in our definition of the error body, only very light assumptions on its structure, namely that it is convex and contains the origin.
Depending on the precise definition, the optimization problem in line~\ref{algline:solveAuxProb} can be more or less complex, or the minimum may not even be attained.
However, if $R$ is a polyhedron, which will be the case in our later applications, then it is even a linear optimization problem.
Additionally, as $R$ contains the origin, the optimization problem in line~\ref{algline:solveAuxProb} of \Cref{algline:solveAuxProb} is always feasible, because $x$ is feasible.
For simplicity of presentation, in what follows, when making statements about the output of \Cref{algo:main}, we always assume that the optimum of the optimization problem in line~\ref{algline:solveAuxProb} is attained.
The results and analysis naturally extend to settings where only an approximate solution is found.

\begin{theorem}\label{thm:roundingGuarantees}
Let $N$ be a finite set, $Q\subseteq \mathbb{R}^N$ a polyhedron, $Z\subseteq Q$, and let $R\subseteq \mathbb{R}^N$ be a convex set containing the origin.
Given $x \in Q$, then using \Cref{algo:main} with a $(Q,Z)$-FPRA with error body $R$ returns a solution $z\in Z$ satisfying
  \begin{itemize}
    \item $c^T z \leq c^T x$, and
    \item $z \in x + (R - R)$.
  \end{itemize}
\end{theorem}
We recall that $x + (R - R)$ is the Minkowski sum of $\{x\}$, $R$, and $-R$.
Hence, $x + (R-R) = \{x + r_1 - r_2\colon r_1, r_2 \in R\}$.
\begin{proof} First, since $y^*-x\in-R$ and $z-y^*\in R$, we immediately obtain that $z=x+(z-y^*)+(y^*-x)\in x+(R-R)$. 

To obtain the cost bound, we crucially utilize that we have an FPRA. A key observation is that since $z$ lies on the minimal face of $Q$ that contains $y^*$, there exists some $\varepsilon > 0$ with
\begin{equation*}
  \overline{y} \coloneqq y^* + \epsilon\cdot (y^* - z) \in Q.
\end{equation*}
This follows because every constraint of $Q$ that is tight at $y^*$ is also tight at $z$. 
Note that because $z$ is the output of a $(Q,Z)$-FPRA with error body $R$ applied to $y^*$, we have
\begin{equation*}
    \overline{y}-y^*=\epsilon\cdot(y^*-z)\in-\epsilon R.
\end{equation*}
Thus,
\begin{equation}\label{eq:yHatWRTX}
 \overline{y}-x=(\overline{y}-y^*)+(y^*-x)\in -(1+\varepsilon) R,
\end{equation}
where we used that $R$ is convex, and so $-\alpha R - \beta R = -(\alpha+\beta)R$ for any $\alpha, \beta \in \mathbb{R}_{\geq 0}$. 

Now, we will use $\overline{y}$ and $x$ to produce a point $\hat y$ lying on the line segment joining $\overline{y}$ and $x$ that also lies in $Q\cap (x-R)$. This will imply that $c^Ty^*\leq c^T\hat y$, and the way in which $\hat y$ is constructed will then yield that $c^Tz\leq c^Tx$. 
Consider the following point (see \Cref{fig:sketch}):
\begin{equation}\label{eq:defYHat}
  \hat{y} = \frac{\varepsilon}{1+\varepsilon} x + \frac{1}{1+\varepsilon} \overline{y},
\end{equation}
which, by using the definition of $\overline{y}$, can be written equivalently as
\begin{equation}\label{eq:hatYAlt}
  \hat{y} = x + \frac{1}{1+\varepsilon} \left(\overline{y} - x\right)
  = y^* + \frac{\varepsilon}{1+\varepsilon} (x - z).
\end{equation}
Because, by \Cref{eq:defYHat}, $\hat{y}$ is a convex combination of $x$ and $\overline{y}$, which both lie in $Q$, we have $\hat{y}\in Q$ by convexity of $Q$.
Moreover, using \Cref{eq:yHatWRTX}, we obtain $\hat y\in x-R$.

Thus, $\hat{y}$ is a feasible solution to the optimization problem in \Cref{algline:solveAuxProb}, which implies $c^T y^* \leq c^T \hat{y}$.
This, in turn, implies the desired cost guarantee $c^T z \leq c^T x$ by \Cref{eq:hatYAlt}.
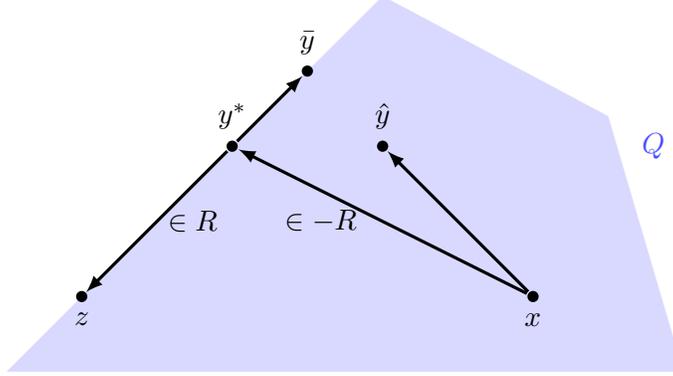
\begin{figure}
    \centering
    
\begin{tikzpicture}[scale=2]

\fill[blue, opacity=0.15] (-3.5,0) -- (-1,2.5) -- (0.5,1.7) --(1,0) -- cycle;
\node[opacity=0.7] (q) at (0.8,1.5) {\textcolor{blue}{$Q$}};

\node[fill=black, circle, inner sep=1.5pt, outer sep=0.5pt,label=below:$x$] (x) at (0,0.5) {};
   \node[fill=black, circle, inner sep=1.5pt, outer sep=0.5pt, label=above:$y^*$] (ystar) at (-2,1.5) {};
   \node[fill=black, circle, inner sep=1.5pt, label=below:$z$] (z) at (-3,0.5) {};
    \node[fill=black, circle, inner sep=1.5pt,  outer sep=0.5pt, label=above:$\bar y$] (ybar) at (-1.5,2) {};
     \node[fill=black, circle, inner sep=1.5pt,  outer sep=0.5pt, label=above:$\hat y$] (yhat) at (-1,1.5) {};

\draw[->,>=latex, very thick] (x) to node[left] {$\in -R \,\,\,$ } (ystar);
    \draw[->,>=latex, very thick] (ystar) to node[above, right] {$\in R$} (z);
    \draw[->,>=latex, very thick] (ystar) to  (ybar);
    \draw[->,>=latex, very thick] (x) to (yhat);

\node(shift) at (0,-0.1) {};
\end{tikzpicture}
     \caption{The relation of vectors and points that appear in the proof of \Cref{thm:roundingGuarantees}.
    We choose $\hat y \coloneqq \frac{\varepsilon}{1+\varepsilon}x + \frac{1}{1+\varepsilon}\bar{y}$ as a convex combination of $\bar y$ and $x$.
    }
    \label{fig:sketch}
\end{figure}
\end{proof}

We briefly justify our earlier comment, that our setup with an explicit error body $R$ generalizes the types of rounding errors considered in \cite{linharesApproximatingMincostChainconstrained2018}. They consider the problem of finding a minimum-cost vertex of a polytope $Q$ that satisfies additional side constraints $Ax\leq b$, and Theorem 5.2 in their paper shows that if one has an FPRA that given $x\in Q$ outputs $z$ (on the minimal face of $Q$ containing $x$) such that $Ax-\Delta\leq Az\leq Ax+\Delta$, then one can obtain an extreme point $\hat x$ of $Q$ of cost at most the optimum, satisfying $A\hat x\leq b+2\Delta$. In our language, this corresponds to having access to a $(Q,Z)$-FPRA with error body $R$, where $Z$ is the set of vertices of $Q$ and $R=\{x \colon -\Delta\leq Ax\leq\Delta\}$. Then, \Cref{thm:roundingGuarantees} can be used to obtain the same guarantee by taking $x$ to be an optimal solution to $\min\{c^Ty \colon y\in Q,\ Ax\leq b\}$. We note that when $R$ is polyhedral, the LP-based approach in \cite{linharesApproximatingMincostChainconstrained2018} can also be utilized to obtain \Cref{thm:roundingGuarantees}; our setup allows one to consider non-polyhedral error bodies as well (e.g., norm balls), and our proof is different and more elementary.

\subsection{Obtaining cost-violation tradeoffs}\label{sec:tradeoff}

We now show that we can achieve a tradeoff between the cost of the rounding result~$z$ and the rounding error if, additionally to the requirements so far, $Q\subseteq \mathbb{R}_{\geq 0}^N$ is in the nonnegative orthant and $c\in \mathbb{R}^N_{\geq 0}$ is also nonnegative.
To this end, we modify \Cref{algo:main} so that instead of optimizing over $Q \cap (x-R)$ we optimize over $Q \cap (x-\lambda R)$ at the start of the algorithm, where $\lambda\in (0,1]$ is some parameter controlling the tradeoff.
(See \Cref{algo:modified} below for the modified algorithm.)

\begin{algorithm2e}[ht]
  \caption{From a $(Q,Z)$-FPRA with error body $R$ to a cost-violation tradeoff rounding algorithm with parameter $\lambda \in (0,1]$. (Recall that we assume $Q\subseteq \mathbb{R}^N_{\geq 0}$ and $c\in \mathbb{R}^N_{\geq 0}$.)}\label{algo:modified}
  Find an optimal solution $y^*$ to $\argmin\{c^T y \colon y \in Q \cap (x - \lambda R) \}$.\label{line:opt_modified}

  Apply the $(Q,Z)$-FPRA to $y^*$ to obtain a solution $z$.

  \Return $z$.
\end{algorithm2e}

Using an analysis analogous to the proof of \Cref{thm:roundingGuarantees}, we obtain the following tradeoff, where the case $\lambda = 1$ corresponds to \Cref{thm:roundingGuarantees}.
Analogous to \Cref{algo:main}, we assume in the following that the optimum of the optimization problem in line~\ref{line:opt_modified} of \Cref{algo:modified} is always attained.

\begin{corollary}
\label{cor:tradeoff}
Let $N$ be a finite set, $c:N\to \mathbb{R}_{\geq 0}$, $\lambda\in (0,1]$, $Q\subseteq \mathbb{R}_{\geq 0}^N$ a polyhedron, $Z\subseteq Q$, and let $R\subseteq \mathbb{R}^N$ be a convex set containing the origin.
Using \Cref{algo:modified} with a $(Q,Z)$-FPRA with error body $R$ returns a solution $z\in Z$ satisfying
    \begin{itemize}
        \item $c^T z \leq \frac{1}{\lambda} \cdot c^Tx$, and
        \item $z \in x +(R - \lambda R)$.
    \end{itemize}
\end{corollary}

\begin{proof}
We follow the lines of the proof of \Cref{thm:roundingGuarantees} with an adapted definition of the point $\hat y$.

As $y^*$ lies on a minimal face of $Q$, there is an $\epsilon >0$ such that 
\[
\bar y \coloneqq y^* + \epsilon (y^* -z) \in Q.
\]
This time we have
\begin{equation}\label{eq:yHatWRTXLambda}
\bar{y} \in x - (\lambda+\epsilon) R,
\end{equation}
which follows from $y^*\in x - \lambda R$, which is an immediate consequence of the definition of $y^*$, and $y^* - z \in -R$, which holds because $z$ is the output of the FPRA with error body $R$.

Now, define 
\[
\hat y \coloneqq \frac{\epsilon}{\lambda+\varepsilon} x + \frac{\lambda}{\lambda + \epsilon} \bar y.
\]
Note that $\hat y \in Q$ because it is a convex combination of points in the convex set $Q$.
Moreover, by \Cref{eq:yHatWRTXLambda}, we have $\hat{y}\in x-\lambda R$.
Hence, $\hat y$ is a feasible solution to the optimization problem in \Cref{line:opt_modified}, and thus $c^T y^* \leq c^T \hat y$.
Finally, some basic computations yield the desired cost guarantee:
\begin{align*}
0 &\leq \frac{\epsilon}{\lambda + \epsilon} \cdot c^T x + \frac{\lambda}{\epsilon + \lambda} c^T\bar y - c^Ty^*\\
&= \frac{\epsilon}{\lambda + \epsilon} \cdot c^Tx + \left( \frac{\lambda}{\lambda+\epsilon} (1+\epsilon) -1\right)c^Ty^* - \frac{\epsilon \lambda}{\lambda+\epsilon} c^Tz\\
&= \frac{1}{\lambda+\varepsilon}\left( \epsilon c^Tx -\underbrace{\epsilon(1-\lambda)}_{\substack{\geq 0 }} \underbrace{c^T y^*}_{\geq 0} - \epsilon \lambda c^T z\right)\\
&\leq \frac{\varepsilon \lambda}{\lambda+\varepsilon}\left( \frac{1}{\lambda} c^Tx  - c^T z \right),
\end{align*}
where the first inequality uses $c^T y^* \leq c^T \hat y$,
the first equality uses the definition of $\bar{y}$,
and the last inequality holds because $\lambda \in (0,1]$ and $c^T y^* \geq 0$ (which holds because $c\in \mathbb{R}^N_{\geq 0}$ and $y^*\in Q \subseteq \mathbb{R}_{\geq 0}$).

\end{proof}
 \section{Applications}\label{sec:applications}

In this section we apply the approach presented in \Cref{sec:approach} to single-source unsplittable flows and the ring loading problem. 
\subsection{Implications for single-source unsplittable flows}

To apply \Cref{thm:roundingGuarantees} to the setting of unsplittable flows, 
we consider the polytope $Q_G$ as defined in \Cref{sec:intro} with $Z_G \subseteq Q_G$ containing all unsplittable flows as defined in~\Cref{eq:defZForUnsplit}. The following proposition follows directly from the fact that non-flow carrying arcs can be deleted before applying  \Cref{conj:morellLowerBounds} as expanded on already in \Cref{sec:approach}.

\begin{proposition}
\label{prop:FPRA}
If \Cref{conj:morellLowerBounds} holds, there is a polynomial time
$(Q_G,Z_G)$-FPRA with error body $R$, where
  \begin{equation*}
    R = [-d_{\max}, d_{\max}]^A.
  \end{equation*}
  Moreover, the algorithm returns an actual unsplittable path flow $\mathcal{P}= (P_t)_{t\in T}$.
\end{proposition}

\Cref{prop:FPRA} and \Cref{thm:roundingGuarantees} directly imply \Cref{thm:mainUnsplitReduction}.
\begin{proof}[Proof of \Cref{thm:mainUnsplitReduction}]
By \Cref{thm:roundingGuarantees} and \Cref{prop:FPRA}, we have that \Cref{algo:main} returns an unsplittable flow $z$ of cost at most $c^T z \leq c^T x$ and satisfying $z \in x + (R - R) = x + [-2 d_{\max}, 2 d_{\max}]$, as desired.
\end{proof}

Moreover, \Cref{prop:FPRA} together with \Cref{cor:tradeoff} directly imply the following corollary. 

\begin{corollary}
    Assuming \Cref{conj:morellLowerBounds}, given a fractional flow $x$ and some $\lambda \in (0,1]$, then there is a polynomial time algorithm which computes an unsplittable flow $\mathcal{P}= (P_t)_{t\in T}$ with
    \begin{itemize}
        \item $c^T f ^\mathcal{P}\leq \frac{1}{\lambda} c^T x$, and
        \item $x(a) - (1+\lambda) d_{\max} \leq f^\mathcal{P}(a) \leq x(a) + (1+\lambda) d_{\max} $ for all $a \in A$.
    \end{itemize}
\end{corollary}

\subsection{Implications for the ring loading problem}

The ring loading problem is a multi-commodity unsplittable flow problem on a cycle. 
A ring loading network $(C=(V,E), (\{s_i,t_i\})_{i \in [k]}, (d_i)_{i \in [k]})$
consists of an undirected cycle $C=(V,E)$, together with $k$ commodities, each consisting of a source-sink pair $\{s_i,t_i\} \in \big(\begin{smallmatrix}V \\ 2\end{smallmatrix}\big)$ and a demand $d_i\in \mathbb{R}_{\geq 0}$ for $i \in [k]$. 
A fractional solution can split the demand of a commodity arbitrarily among the two $s_i$-$t_i$ paths in the cycle, while an unsplittable solution chooses one of these two paths for each commodity and routes the whole demand of that commodity along the chosen path.

For $i\in [k]$, we denote by $P_i^1\subseteq E$ and $P_i^2\subseteq E$ the two different $s_i$-$t_i$ paths in the cycle $C$.
For a fractional solution to the ring loading problem we consider the load vector
\begin{equation*}
	Q_C \coloneqq \left\{x \in \mathbb{R}_{\geq 0}^E \colon \exists \lambda_i\in [0,1] \text{ for } i\in [k] \text{ s.t. } x(e) = \sum_{\substack{i \in [k]:\\ e \in P_i^1}} \lambda_i d_i  + \sum_{\substack{i\in [k]:\\ e\in P_i^2}} (1-\lambda_i) d_i\text{ for all } e \in E\right\}.
\end{equation*}
The value of a fractional solution is the maximum load
$\max_{e \in E} x_e$.
An unsplittable solution is a selection of $s_i$-$t_i$ paths $\mathcal{P}=(P_i)_{i \in [k]}$, i.e., $P_i\in \{P_i^1, P_i^2\}$ for $i\in [k]$.
Let the load of edge $e$ be defined as
\begin{align*}
    \load_{e}(\mathcal{P}) &\coloneqq \sum_{i \in [k]: e\in P_i} d_i.
\end{align*}
The objective in the ring loading problem is to minimize the maximal load on an edge.

We will consider a weighted and capacitated variant of the problem, where we are given a non-negative cost function $c \colon E \to \mathbb{R}_{\geq 0}$ and edge capacities $u(e)$ for all $e \in E$.
A ring loading network, together with a cost function and capacities yields a \emph{weighted, non-uniform ring loading instance}. Here, the objective is to compute a solution such that the load on each edge is at most its capacity, i.e., $x_e \leq u(e)$ or $\load_e(\mathcal{P}) \leq u(e)$, while minimizing the cost $c^Tx$ or $c^T\load (\mathcal{P})$.
If, additionally, all capacities $u(e)$ for $e \in E$ are the same, then we denote the instance as a \emph{weighted, uniform ring loading instance}, and use $u_{\mathrm{unif}}$ to denote the uniform capacity.
The unweighted version of these problems considers capacity bounds but no costs.

Däubel~\cite{daubel2022improved} showed that given a fractional solution $x$, there is a polynomial time algorithm which computes an unsplittable solution $\mathcal{P}$ with $\load_e(\mathcal{P}) \leq x_e + \frac{13}{10}d_{\max}$ for all $e \in E$, where $d_{\max}= \max_{i \in [k]} d_i$.
Moreover, Skutella~\cite{skutella2016note} showed that the additive violation with respect to a fractional solution has to be at least $\frac{11}{10} d_{\max}$.
Analogous to the single-source unsplittable flow setting, an algorithm as presented by Däubel~\cite{daubel2022improved} can be applied to the unweighted non-uniform ring loading problem by computing a fractional solution $x$ to the non-uniform ring loading instance and applying the algorithm to that solution, to compute an unsplittable solution $\mathcal{P}$ with $\load_e(\mathcal{P}) \leq x_e+\frac{13}{10}d_{\max} \leq u(e)+\frac{13}{10}d_{\max}$.

We first show in the following lemma that any algorithm for the uniform ring loading problem can be transformed in a black-box way in an algorithm for the non-uniform ring loading problem, while preserving the additive violation guarantees of the capacities.

\begin{lemma}
\label{lem:nonuniform_as_hard_as_uniform}
If for every fractionally feasible instance of the uniform ring loading problem, there exists an unsplittable solution $\mathcal{P}$ with guarantee $\load_e(\mathcal{P}) \leq u_{\mathrm{unif}} + \alpha d_{\max}$ for all $e \in E$, then also for every fractionally feasible instance of the non-uniform ring loading problem, there exists an unsplittable solution $\mathcal{P}$ with  $\load_e(\mathcal{P}) \leq u(e) + \alpha d_{\max}$ for all $e \in E$.

If such an unsplittable flow for the uniform ring loading problem can be computed in polynomial time, then this also applies to the non-uniform ring loading problem.
\end{lemma}
\begin{proof}
    We consider a non-uniform instance with capacities $u(e)$ for $e \in E$ and show how to solve this with a guarantee of $\alpha$, given an algorithm with this guarantee for the uniform setting.
    
    First, we transform our instance into a uniform instance with global capacity $u_{\mathrm{unif}}$.
    For this sake, let $u_{\mathrm{unif}} \coloneqq \max_{e \in E} u(e)$. 
    For every edge $e=\{v,w\}$ with $u(e) \neq u_{\mathrm{unif}}$, we subdivide $e$ into two edges $e_1=\{v,w_{\text{new}} \}$ and $e_2=\{w_{\text{new}},w \}$.
    We define $r_e \coloneqq \left\lceil\frac{u_{\mathrm{unif}}-u(e)}{d_{\max}}\right\rceil$.
    Then we introduce a set $A(e)$ of $2r_e$ new commodities, each with demand $\frac{u_{\mathrm{unif}}-u(e)}{r_e} \leq d_{\max}$. 
    For $r_e$ of the new commodities, their start and end node are the endpoints of $e_1$.
    For the other $r_e$ new commodities, their start and end node are the endpoints of $e_2$.

    If the non-uniform instance admits an unsplittable solution with capacity violation of at most $\alpha d_{\max}$, then such a solution still exists for the new, uniform instance. This follows as any solution to the non-uniform instance can be extended to the uniform setting without increasing the capacity violation by routing each artificial commodities along a single edge (either $e_1$ or $e_2$).

    To this new, uniform instance, we can apply our algorithm with guarantee $\alpha$.
    So assume, we obtained an unsplittable solution with a capacity violation of at most $\alpha d_{\max}$ for the uniform setting. We claim that by removing the artificial commodities, we obtain a solution to the original instance with at most the same capacity violation.
    This follows as for each edge $e$, there is either on $e_1$ or on $e_2$ a load of at least $u_{\mathrm{unif}}-u(e)$ induced by the artificial commodities $A(e)$.
    To see this, observe that each artificial commodity in $A(e)$ induces a load of $\frac{u_{\mathrm{unif}}-u(e)}{r_e}$ on one of the edges $e_1$ and $e_2$ (on which one depends on the chosen path for this commodity).
    Thus, on one of the two edges $e_1$ and $e_2$, at least $r_e$ of the $2r_e$ commodities in $A(e)$ induce such a load, leading to a total load of at least $u_{\mathrm{unif}}-u(e)$ induced by the  commodities $A(e)$.
\end{proof}

We will show that
\Cref{cor:tradeoff} implies the following. 

\begin{theorem}
	\label{thm:ring_loading}
	Assume we can compute, for every fractionally feasible instance of the unweighted ring loading problem, an unsplittable solution $\mathcal{P}$ with guarantee $\load_e(\mathcal{P}) \leq u_{\mathrm{unif}} + \alpha d_{\max}$ for all $e \in E$. Then, given a weighted ring loading instance together with a fractional solution $x$ and some $\lambda \in (0,1]$, there is an algorithm that computes an unsplittable solution $\mathcal{P}=(P_i)_{i \in [k]}$, s.t. 
\begin{itemize}
\item $\sum_{i=1}^k c(P_i) \leq \frac{1}{\lambda} c^T x  $, and
	\item $\load_e(\mathcal{P}) \leq x_e+ (1+\lambda) \alpha \cdot d_{\max} $ for all $e \in E$.
\end{itemize}
If the algorithm for the unweighted ring loading problem has polynomial running time, then so does the algorithm for the weighted ring loading problem.
\end{theorem}

\begin{proof}
To utilize \Cref{cor:tradeoff}, we need a $(Q_C,Z)$-FPRA, where $Z$ corresponds to the load vectors of unsplittable solutions, with an error body that encodes the $\alpha d_{\max}$ two-sided additive error. However, the theorem statement only assumes that we can obtain one-sided error bounds for uniform-capacity instances. By \Cref{lem:nonuniform_as_hard_as_uniform}, the restriction to uniform capacities is without loss of generality. We make a series of transformations to the instance that will enable us to argue that the one-sided guarantee assumed in the theorem statement suffices to yield two-sided error bounds for the equivalent transformed instance, so that we can then apply \Cref{cor:tradeoff}.

We first use an observation made in \cite{schrijver1999ring,skutella2016note} to restrict our attention to more structured instances. 
We repeat the argument here for completeness.

First, observe that if any demand is sent unsplit in the fractional solution, i.e., all the demand $d_i$ is routed along a single $s_i$-$t_i$ path, we can as well, fix this path and delete the commodity and reduce the load along the path by $d_i$. 

 Say that two commodities $i$ and $j$ are in \emph{parallel} if there is an $s_i$-$t_i$ path and an $s_j$-$t_j$ path which are edge-disjoint. Otherwise, we say $i$ and $j$ \emph{cross}.

\begin{figure}
    \centering
 \begin{tikzpicture}[scale=1.2, every node/.style={font=\small}]
    \def\r{1.5}  \def\raa{1.4} \def\rab{1.85} \def\dist{4.5}

        \begin{scope}
\draw[thick] (0,0) circle(\r);
            
\def\sOneAngle{135}
            \def\tOneAngle{45}
            \def\sTwoAngle{-45}
            \def\tTwoAngle{-135}
            
\path (0,0) ++(\sOneAngle:\r) coordinate (s1);
            \path (0,0) ++(\tOneAngle:\r) coordinate (t1);
            \path (0,0) ++(\sTwoAngle:\r) coordinate (s2);
            \path (0,0) ++(\tTwoAngle:\r) coordinate (t2);
            
\fill[blue] (s1) circle (2pt);
            \node[blue, shift={(10:-10pt)}] at (s1) {$s_1$};
            
            \fill[blue] (t1) circle (2pt);
            \node[blue, shift={(15:10pt)}] at (t1) {$t_1$};
            
            \fill[red] (s2) circle (2pt);
            \node[red, shift={(-15:-10pt)}] at (s2) {$s_2$};
            
            \fill[red] (t2) circle (2pt);
            \node[red, shift={(10:10pt)}] at (t2) {$t_2$};
            
\draw[very thick, blue] 
                ([shift=(\sOneAngle:\raa)]0,0) arc[start angle=\sOneAngle, end angle=\tOneAngle, radius=\raa];
                
            \draw[very thick, red] 
                ([shift=(\sTwoAngle:\rab)]0,0) arc[start angle=\sTwoAngle, end angle=\tTwoAngle, radius=\rab];
\end{scope}

\begin{scope}[shift={(\dist,0)}]
\draw[thick] (0,0) circle(\r);
            
\def\sOneAngle{135}
            \def\sTwoAngle{45}
            \def\tOneAngle{-45}
            \def\tTwoAngle{-135}
            
\path (0,0) ++(\sOneAngle:\r) coordinate (s1);
            \path (0,0) ++(\sTwoAngle:\r) coordinate (s2);
            \path (0,0) ++(\tOneAngle:\r) coordinate (t1);
            \path (0,0) ++(\tTwoAngle:\r) coordinate (t2);
            
\fill[blue] (s1) circle (2pt);
            \node[blue, shift={(10:-10pt)}] at (s1) {$s_1$};
            
            \fill[blue] (t1) circle (2pt);
            \node[blue, shift={(15:10pt)}] at (t1) {$t_1$};
            
            \fill[red] (s2) circle (2pt);
            \node[red, shift={(10:8pt)}] at (s2) {$s_2$};
            
            \fill[red] (t2) circle (2pt);
            \node[red, shift={(10:10pt)}] at (t2) {$t_2$};
            
\draw[very thick, blue] 
                ([shift=(\sOneAngle:\raa)]0,0) arc[start angle=\sOneAngle, end angle=\tOneAngle, radius=\raa];
                
            \draw[very thick, red] 
                ([shift=(\sTwoAngle:\rab)]0,0) arc[start angle=\sTwoAngle, end angle=\tTwoAngle, radius=\rab];
\end{scope}

\begin{scope}[shift={(2*\dist,0)}]
\def\r{1.5cm}

\foreach \angle/\name/\label in {
        170/s1/$s_1$, 
        150/s2/$s_2$, 
        130/s3/$s_3$, 
         10/s4/$s_k$
    } {
        \path (0,0) ++(\angle:\r) coordinate (\name);
        \fill (\name) circle (2pt);
        \node at (\name) [label={\angle:\label}] {};
    }
        \node at (90:\r)  [label={$\ldots$}] {};

\foreach \angle/\name/\label in {
        -170/tk/$t_k$, 
        -150/tk1/$t_k-1$, 
        -50/t3/$t_3$, 
         -30/t2/$t_2$, 
         -10/t1/$t_1$
    } {
        \path (0,0) ++(\angle:\r) coordinate (\name);
        \fill (\name) circle (2pt);
        \node at (\name) [label={\angle:\label}] {};
    }
        \node at (-90:\r)  [label={$\ldots$}] {};

\draw[thick] (0,0) circle(\r);  

\draw[very thick, blue] 
  ([shift=(130:\r)]0,0) arc[start angle=130, end angle=150, radius=\r];
\draw[very thick, blue] 
  ([shift=(-30:\r)]0,0) arc[start angle=-30, end angle=-50, radius=\r];
\end{scope}

\end{tikzpicture}
     \caption{On the left we see two parallel commodities, in the middle two crossing commodities and on the right we see a reduced instance, where two edges whose loads sums up to the total demand are highlighted in blue. For all instances the undirected graph $G=(V,E)$ was drawn.}
    \label{fig:enter-label}
\end{figure}
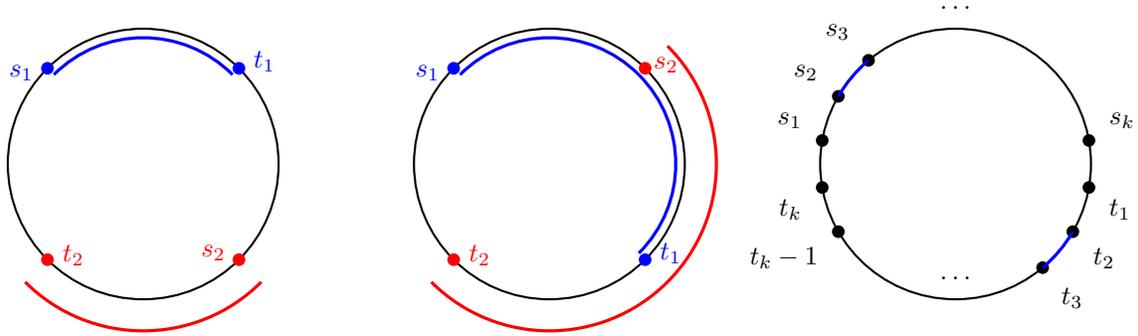

Schrijver, Seymour, and Winkler~\cite{schrijver1999ring} and Skutella~\cite{skutella2016note} observed that given an instance with parallel demands and a fractional solution, then there exists a fractional solution where one of the commodities just uses a single path while it does not increase the load on any edge. 
Let $i$ and $j$ be two parallel commodities.
For $\ell \in \{i,j\}$, let $P_\ell^1$ and $P_\ell^2$ be the two paths of commodity $\ell$ along which the fractional solution routes a load of $\lambda_\ell^1$ and $\lambda_\ell^2$, respectively. 
Up to renaming, we have $P_i^1 \subset P_j^2$ and $P_j^1 \subset P_i^2$ because $i$ and $j$ are parallel; see also \Cref{fig:enter-label}. 
By shifting $\min \{ \lambda_i^2, \lambda_j^2\}$ from $P_i^2$ to $P_i^1$ and from $P_j^2$ to $P_j^1$ we obtain a fractional solution that does not increase the load on any edge. Moreover, the commodity by which the minimum is attained is routed unsplittably after the modification and can thus be eliminated, as explained before.
As the load on edges does not increase and costs are non-negative, also the cost does not increase by this operation.

We can iterate this process to obtain a reduced instance of the ring loading problem where all commodities are crossing together with a fractional solution $\overline x$. Let $F$ be the set of indices of commodities whose paths have been fixed and which have been deleted during this process and let $P_i$ be the path for $i \in F$. Then, for all $e \in E$, we have that $\overline x (e) + \sum_{i \in F: e \in P_i} d_i\leq x(e)$. 
Let $k'$ be the number of remaining commodities and relabel them from 1 to $k'$.

Like in \cite{skutella2016note}, we observe that we can assume that the vertices around the cycle in clockwise order are $s_1, s_2, \ldots, s_{k'}, t_1, \ldots,  t_{k'}$.
This holds as all commodities are crossing, and in the original undirected graph $s_i$ and $t_i$ can be used interchangeably.
Moreover, if a vertex does not belong to any demand pair, we can contract the two adjacent edges $e_1$ and $e_2$ and assign them the cost $c(e_1)+c(e_2)$ and the minimum of the two capacities because the load on those edges is the same in any (fractional or unsplittable) solution.

We continue by observing that any fractional solution $y$ for such an instance satisfies $y_{\{s_i,s_{i+1}\}}+y_{\{t_i,t_{i+1}\}} = \sum_{j \in [k']} d_j$ for any $i \in [k']$ (where we interpret $s_{k'+1} =t_1$ and $t_{k'+1}=s_1$).
This holds because $\{s_i,s_{i+1}\}$ carries the load of all commodities $j \in [i]$ that is sent clockwise and the edge $\{t_{i+1},t_i\}$ carries the load of all commodities $j \in [i]$ that is sent counter-clockwise. Similarly, $\{s_{i+1},s_{i}\}$ carries the load of all commodities $j \in \{i+1, \ldots , k'\}$ that is sent counter-clockwise and $\{t_i,t_{i+1}\}$ carries the load of all commodities $j \in \{i+1, \ldots , k'\}$  that is sent clockwise. In total this shows 
\begin{equation}
	\label{eq:opposing_edges}
y_{\{s_i,s_{i+1}\}}+y_{\{t_i,t_{i+1}\}} = \sum_{j \in [k']} d_j \qquad \text{for all } i \in [k'].
\end{equation}

We exploit this to observe that any unsplittable solution $\overline{\mathcal{P}}$ with 
$\load_e(\overline{\mathcal{P}}) \leq \overline{x}_e +\alpha d_{\max}$ for all edges $e \in E$, also fulfills $\load_e(\overline{\mathcal{P}}) \geq \overline{x}_e -\alpha \cdot  d_{\max}$ for all edges $e \in E$. This follows by applying \eqref{eq:opposing_edges} twice,
once in the first equality below to the unsplittable solution $\overline{P}$ (which is a special case of a fractional solution), and the second time in the second equality below to the fractional solution $\overline{x}$. First consider $e=\{t_i,t_{i+1}\}$:
\begin{align*}
	\load_{\{t_i,t_{i+1}\}}(\overline{\mathcal{P}})&= \left(\sum_{j \in [k']} d_j\right) - 	\load_{\{s_i,s_{i+1}\}}(\overline{\mathcal{P}})\\
	                                               &= \overline x_{\{s_i,s_{i+1}\}}+\overline x_{\{t_i,t_{i+1}\}}- 	\load_{\{s_i,s_{i+1}\}}(\overline{\mathcal{P}})\\
	& \geq \overline x_{\{t_i,t_{i+1}\}} - \alpha \cdot d_{\max}.
\end{align*}
The analogous statement for arc $\{s_i,s_{i+1}\}$, i.e.,  $\load_{\{s_i,s_{i+1}\}}(\overline{\mathcal{P}})\geq \overline x_{\{s_i,s_{i+1}\}} - \alpha \cdot d_{\max}$ follows in a completely symmetric manner. Moreover, by assumption and as we interpret $s_{k'+1} =t_1$ and $t_{k'+1}=s_1$, all arcs are either of the type $\{s_i,s_{i+1}\}$ or $\{t_i,t_{i+1}\}$.

We can assume that we are given a rounding algorithm for the uniform setting with guarantee~$\alpha$. By \Cref{lem:nonuniform_as_hard_as_uniform}, this implies that there is a rounding algorithm for the non-uniform setting with guarantee $\alpha$. By setting $u(e)=\overline{x}(e)$, this yields a $(Q_C,Z_C)$-FPRA with error body
$R=[- \alpha\cdot  d_{\max},\alpha \cdot  d_{\max}]^A$ where $Z_C=\{ (\load_a(\mathcal{P}))_{a \in A} \colon \text{$\mathcal{P} $ is an unsplittable solution}\}$.
Note that any rounding algorithm will be face preserving w.r.t.~$Q_C$, because 
after the pre-processing all demands are crossing and are sent fractionally. Thus, every commodity uses both possible paths. This means, arcs with load 0 cannot be used by any feasible solution and the unsplittable solution will lie on the same face of $Q_C$ as $\overline x$.

This means we can apply \Cref{cor:tradeoff} 
to obtain an unsplittable solution $\widetilde{\mathcal{P}}$ with cost at most $\frac{1}{\lambda} c^T \overline x$ and $\overline x_e - (1+\lambda) \alpha d_{\max} \leq \load_e(\widetilde{\mathcal{P}}) \leq \overline x_e + (1+\lambda) \alpha d_{\max} $.

It remains to add the fixed paths. This will not increase the cost of the unsplittable solution more than the cost of the fractional solution. The same holds true for the load on an edge. Thus, the statement follows. More formally, let $\mathcal{P}$ be $\widetilde{\mathcal{P}}$ together with the fixed paths of commodities $i \in F$. 
We have that 
\begin{align*}
    \load_e(\mathcal{P}) &= \load_e(\widetilde{\mathcal{P}}) + \sum_{i \in F: e \in \widetilde{P}_i} d_i\\
&\leq \overline x_e + (1+\lambda) \alpha d_{\max} + \sum_{i \in F: e \in \widetilde{P}_i} d_i \leq x_e + (1+\lambda) \alpha d_{\max}.
\end{align*}

Similarly, 
\[
    \sum_{i=1}^k  c( P_i ) = \sum_{i=1}^{k'} c( \widetilde P_i ) + \sum_{i \in F} c(P_i) \leq \frac{1}{\lambda}  c^T\overline x + \sum_{i \in F} c(P_i) \leq \frac{1}{\lambda}  c^T x.
\]

Note, however, that due to the rerouting of parallel commodities, the lower bounds on the load with respect to $\overline x$ do not translate to lower bounds with respect to $x$. 
\end{proof}

We remark that the above result for $\lambda=1$ could also be achieved similarly to the cost result for planar graphs in \cite{traubSingleSourceUnsplittableFlows2024}, by applying a standard argument from discrepancy theory previously used in \cite{bansal2022flow}. First, as above, the instance has to be pre-processed such that all commodities are crossing. Then, 
by repeatedly solving the unweighted rounding problem for carefully constructed half-integral solutions, one can solve the weighted rounding problem, again with an additive error that is twice as large as in the unweighted rounding problem. 
One advantage of our technique is that we need only a single call to the algorithm for unweighted instances.

Finally, applying \Cref{thm:ring_loading} to the result of Däubel~\cite{daubel2022improved} for $\lambda=1$ implies \Cref{thm:mainRingLoading}.

\printbibliography

\end{document}